\documentclass{iopart}
\usepackage{iopams}
\usepackage{bbm}
\usepackage{enumerate}
\usepackage{hyperref}
\usepackage{color}
\usepackage{graphicx}

\newtheorem{dfn}{Definition}

\newtheorem{example}[dfn]{Example}

\newtheorem{thm}[dfn]{Theorem}

\def\scaled{\let\onleft=\left\let\onright=\right}
\def\unscale{\let\onleft=\relax\let\onright=\relax}
\unscale
\newcommand{\hilbert}[1]{\ensuremath{\mathfrak{#1}}}

\newcommand{\bmaps}{\ensuremath{\mathfrak B}}

\newcommand{\bbone}{\ensuremath{\mathbbm 1}}

\newcommand{\defeq}{\ensuremath{:=}}
\newcommand{\setdef}{\ensuremath{\;\vert\;}}

\newcommand{\cmpl}{\ensuremath{'}}

\def\scaled{\let\onleft=\left\let\onright=\right}
\def\unscale{\let\onleft=\relax\let\onright=\relax}

\newcommand{\set}[1]{\ensuremath{\onleft\{ #1\onright\}}\unscale}

\newcommand{\abs}[1]{\ensuremath{\onleft\vert #1 \onright\vert}}

\newcommand{\complexes}{\ensuremath{\mathbb C}}
\newcommand{\reals}{\ensuremath{\mathbb R}}

\newcommand{\borel}{\ensuremath{\mathrm{B}}}

\begin{document}

\title{Non-signaling boxes and quantum logics}
\author{T I Tylec and M Ku\'s}
\address{Center for Theoretical Physics,
    Polish Academy of Sciences,
    Aleja Lotnik\'ow 32/46,
    02-668 Warsaw, Poland}
\ead{tylec@cft.edu.pl}

\begin{abstract}
We analyze the structure of the so called non-signaling theories respecting
relativistic causality but allowing correlations violating bounds imposed by
quantum mechanics such as CHSH inequality. We discuss relations among such
theories, quantum mechanics, and classical physics. In particular we
reconstruct the probability theory adequate for the simplest instance of a
non-signaling theory, the two non-signaling boxes world, and exhibit its
differences in comparison with classical and quantum probabilities. We show
that the question whether such a theory can be treated as a kind of
``generalization'' of the quantum theory of the two-qubit system cannot be
answered positively. Some of its features put it closer to the quantum world,
for example measurements must be destructive, on the other hand the Heisenberg
uncertainty relations are not satisfied. Another interesting property
contrasting it from quantum mechanics is that the subset of ``classically
correlated states'', i.e.\ the states with only classical correlations, does
not reproduce the classical world of two two-state systems. Our results
establish a new link between quantum information theory and the well-developed
theory of quantum logics and can shed new light on the problem why quantum
mechanics is distinguished among non-signaling theories.
\end{abstract}




\maketitle

\section{Introduction}

The theory of so-called non-signaling boxes, also known as Box World theory,
became recently popular in the quantum information theory\footnote{The
literature in this topic is extensive, let us only mention that the pioneering
work of Popescu and Rohrlich\cite{popescu1994quantum} has about 634 citations,
incl. 21 in Nature (according to Google Scholar).}. Application of the
non-signaling theories range from discussing fundamental properties of nature
\cite{popescu1994quantum,Pawlowski:2009aa} to proving security of ciphering
protocols \cite{PhysRevLett.97.120405} or finding bounds on communication
complexity \cite{PhysRevLett.96.250401}.

In this work we will focus on the simplest example of non-signaling boxes, a
system defined in Ref.~\cite{barrett2005nonlocal} (see
Fig.~\ref{fig:2-2-box-world}, cf.\ also Ref.~\cite{short2010strong}), composed
of two spatially separated boxes, each having one binary input (with values
denoted by $\set{x, y}$) and one binary output (with values denoted by $\set{0,
1}$). 
The model is supposed to describe the most elementary system composed of
two separated subsystems. We can think of inputs as observables that we choose
to measure and outputs as results of measurements.

\begin{figure}[h]
\begin{center}
    \label{fig:2-2-box-world}
    \includegraphics[width=10cm]{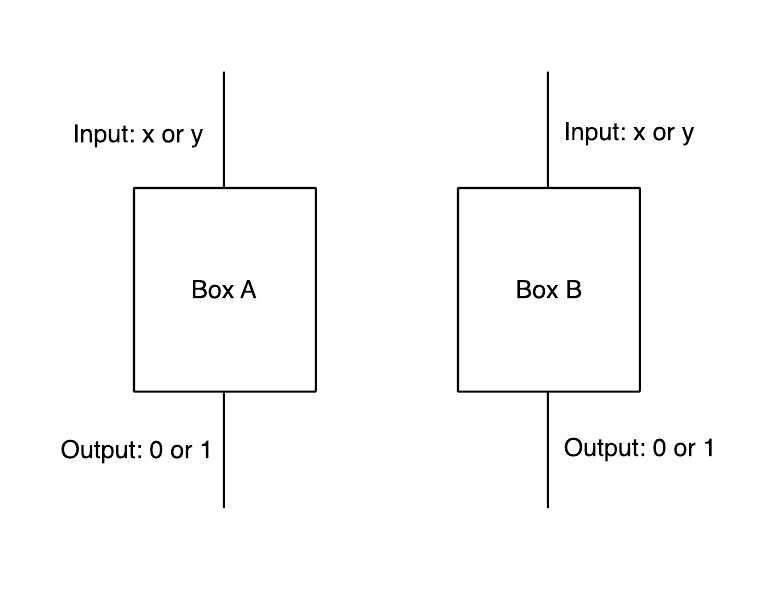}
    \caption{Schematic view of discussed example of non-signaling boxes.
        Output depends on input but in non-deterministic way.}
\end{center}
\end{figure}

Performing multiple measurements we will obtain a sequence of outcomes allowing
to determine relative frequencies $P(\alpha \beta | ab)$ of getting any pair of
outputs $\alpha \beta \in \set{0, 1}\times \set{0, 1}$, given any pair of
inputs $ab \in \set{x, y} \times \set{x, y}$ that can be explained when $P$ is
treated as a probability distribution determined by the actual state of the
system.

Let us thus define a state of a system as $P\colon\set{0,
1}\times\set{0,1}\times\set{x,y}\times\set{x, y}\rightarrow\mathbb{R}$
fulfilling
\begin{enumerate}
            \item[P1] \ $0 \le P(\alpha \beta| ab) \le 1$ (positivity)
            \item[P2] \ $\sum_{\alpha \beta} P(\alpha\beta | ab) = 1$
                (normalization)
            \item[P3] \ $\sum_\alpha P(\alpha \beta | ab) = \sum_\alpha
                P(\alpha \beta | c b)$ and $\sum_\beta P(\alpha \beta | ab)
                = \sum_\beta P(\alpha \beta | a c)$ (non-signaling
                condition),
        \end{enumerate}
where we always sum over the whole domain and all unbound variables are
universally quantified. The last property, non-signaling, is supposed to encode
the principle of relativistic causality, i.e.\ ``what happens in one box does
not influence the other'' \cite{popescu1994quantum} obeyed by spatially
separated subsystems. In the sequel, we will refer to this particular example
of non-signaling boxes as $(2, 2)$-box world.

Our aim is to investigate whether we can actually call $P$ \emph{a
probability}. If yes we can denote by:
\begin{equation}\label{eq:correlation}
    \langle ab \rangle =
    \sum_{\alpha,\beta\in\set{0,1}} (-1)^{\alpha\beta} P(\alpha\beta|ab),
\end{equation}
a quantity that we can interpret as the mean value of the ``observable
$(1-2a)(1-2b)$''. It is called in literature \emph{a correlation}. Then the
following CHSH \cite{clauser1969proposed}-like inequality holds,
\begin{equation}\label{eq:chsh}
    \abs{\langle xx \rangle + \langle xy \rangle + \langle yx \rangle - \langle
    yy \rangle} \le 4,
\end{equation}
since the absolute value of each term on the left hand side does not exceed 1.
In particular for the state described by
\begin{equation}
    P(\alpha\beta|ab) = \bordermatrix{
        \mathrm{} &  xx &  xy &  yx &  yy \cr
             00 & 1/2 & 1/2 & 1/2 &   0 \cr
             01 &   0 &   0 &   0 & 1/2\cr
             10 &   0 & 0   & 0   & 1/2\cr
             11 & 1/2 & 1/2 & 1/2 &   0
    },
\end{equation}
the maximal value of $4$ is obtained. We recall, that if in \eref{eq:chsh} one
substitute for $\langle ab \rangle$ the expectation value given by quantum
mechanics (i.e.\ $\Tr \rho ab$ for some state $\rho$), then the right hand side
of the inequality is $2\sqrt2$ (we will refer to this number as the Tsirelson's
bound). When $\langle ab \rangle$ is an expectation value in classical physics
(i.e.\ $\int_\Gamma a(x) b(x) p(x) dx$ for some probability measure $p(x)dx$ on
the phase space $\Gamma$) then the upper bound of the left hand side of
\eref{eq:chsh} equals $2$\footnote{Although we mention CHSH-like inequalities
(and, in particular, violation of Tsirelson bound) several times, being
acquainted with them is not essential for understanding this paper. Thus for
sake of compactness, we do not provide introduction to this topic. Reader
unfamiliar with this notion can easily find many references. Among multitude of
choices we recommend very interesting book of Streater\cite{streater2007lost},
or shorter paper by Griffiths\cite{griffiths2011quantum}.}. The fact that the
left hand side of \eref{eq:chsh} can be greater than $2\sqrt2$ led to the
conclusion \cite{popescu1994quantum, barrett2005nonlocal} that there are
theories respecting relativistic causality, but still exhibiting ``stronger
correlations'' than quantum mechanics\footnote{It is worth to point out that
the problem of violating $2\sqrt2$ bound in \eref{eq:chsh} was discussed in the
algebraic setting by Landau in
\cite{landau1992experimental,landau1987experimental} already in 1992. In a
very general setting of Segal's axiomatization \cite{segal1947postulates},
Landau showed that the violation of Tsirelson bound implies non-distributivity
of the algebra of observables. The only known examples of non-distributive
Segal's algebras \cite{sherman1956segal} are not suitable for description of
two causally separated systems. Thus, it seems that in the algebraic setting,
the relativistic causality is enough to single out quantum mechanics among
other reasonable theories. The only possible exception could be provided by
some yet unknown non-distributive Segal algebras.}.

In classical and quantum physics the interpretation of a state in terms of
probabilities is well established and does not pose any conceptual problems.
In the case of non-signaling boxes, however, a more careful analysis is needed.
Namely, a state $P$ is in fact only a function that satisfies P1--P3. The only
reason to call it probability is motivated by a \emph{thought experiment} in
which one imagine that having such boxes, performing measurements and writing
down the results one will obtain sequence of outcomes that can be explained
when $P$ is treated as a probability distribution. But there are no physical
non-signaling boxes and in our opinion interpretation of mathematical objects
used in the theory cannot be based on thought experiments. If the
interpretation of $P$ as a probability is not justified then so is drawing any
conclusions from \eref{eq:chsh}. The meaning of non-signaling is not clear
either.

In the following, exploiting the framework of quantum logics (see e.g.\
Ref.~\cite{ptak1991orthomodular}) we were able to show by construction that the
logic of propositions in $(2, 2)$-box world is an orthomodular poset (all
relevant definitions are given in the next section). This is sufficient to
interpret $P$ as a generalized {probability} \cite{pulmanova2007}. Moreover, by
showing that propositions that can be identified with propositions about only
one of the boxes are pairwise compatible we prove that (2,2)-box world can
indeed describe spatially separated subsystems 
(it seems that in general the non-signalling condition is weaker restriction
than the compatibility).

The paper is organized in the following way. For Reader's convenience, we begin
with a short introduction to the framework of quantum logics in
Section~\ref{sec:prob}. In Section~\ref{sec:logics-of-bwo} we construct the
logic of (2,2)-box world and analyze its basic structural properties.

In sections \ref{sec:box-world-one} and \ref{sec:restr-class-case} we compare
the obtained probability theory with classical and quantum one. We make
interesting observation that although from the mathematical point of view a
quantum logic is a more general structure than an orthomodular lattice (the
structure that is behind quantum probability) or Boolean algebra (that
describes classical probability), the simple statement that the theory of
non-signaling boxes is more general than quantum mechanics is not justified.
This stems from the fact, that the logic of even the simplest quantum model (a
qubit), has infinite number of propositions, contrary to 82 propositions of
$(2, 2)$-box world. From that point of view the logic of $(2, 2)$-box world is
much more similar to the logic of an analogous classical system. In fact, we
show in Section~\ref{sec:restr-class-case} that $(2, 2)$-box world logic can be
represented by wisely chosen subsets of the phase space of a classical two box
system. This has profound consequences. In particular, an analog of Heisenberg
uncertainty relations does not hold in $(2, 2)$-box world. 

In parallel, we study properties of $(2, 2)$-box world with the set of states
restricted to the so-called ``classically correlated boxes''. We show that the
resulting logical structure, contrary to the intuition, is exactly the same as
the logical structure of unrestricted $(2, 2)$-box world. Consequently,
``classically correlated boxes'' are not ``classical''.

In Section \ref{sec:embedding-box-world} we show how the so-called
``non-local'' states are responsible for peculiarities of (2,2)-box world
logic. Final remarks can be found in Section~\ref{sec:conclusions}.

All computations where performed using Wolfram Mathematica. For a reference the
notebook can be found in a \verb$git$ repository: see Ref.~\cite{gitrepo}. We
emphasize that we rely on the exact symbolic methods provided by Mathematica.
Presented results \emph{are not} based on any numerical calculations. Finally,
whenever we use ``clearly'', ``easily follows'', etc.\ we mean that the fact is
either obvious or can be easily verified using a computer, i.e.\ involves a lot
of combinatorics or exhaustive check over some large, but finite set.

\section{Framework of quantum logics}
\label{sec:prob}

For reference, we provide here a very concise introduction to quantum logics.
We remind only these definitions and facts that will be used in the sequel.
Interested reader can find much more detailed exposition in
\cite{ptak1991orthomodular} or in \cite{pulmanova2007}.


The primary object of our interest will be a partially ordered set (poset),
i.e.\ a set equipped with a reflexive, antisymmetric and transitive relation.

\begin{dfn}[see Def.~1.1.1 in \cite{ptak1991orthomodular}]
    A \emph{quantum logic} is a partially ordered set $\mathcal L$
    with a map $\cmpl\colon\mathcal L\to \mathcal L$ such that
    \begin{enumerate}
        \item[L1] there exists the greatest (denoted by $\bbone$)
            and the least (denoted by $0$) element in $\mathcal L$,
        \item[L2] map $a \mapsto a\cmpl$ is order reversing, i.e.\
            $a \le b$ implies that $b\cmpl \le a\cmpl$,
        \item[L3] map $a \mapsto a\cmpl$ is idempotent, i.e.\
            $(a\cmpl)\cmpl = a$,
        \item[L4] for a countable family $\set{a_i}$, s.t.\ $a_i \le a_j\cmpl$
            for $i\neq j$, the supremum $\bigvee \set{a_i}$ exists,
        \item[L5] if $a\le b$ then $a\vee(b\wedge a\cmpl)$ exists and
            $b = a\vee(b\wedge a\cmpl)$ (orthomodular law),
    \end{enumerate}
    where
    $a \vee b$ is the least upper bound
    and $a \wedge b$ the greatest lower bound of $a$ and $b$.
    \label{thm:qlogic}
\end{dfn}

\begin{dfn}[see Def.~2.4.1 in \cite{ptak1991orthomodular}]
An element $a\in\mathcal{L}$ is called an \emph{atom} if $a\ne 0$ and for
every $b\le a$ either $b= 0$ or $b =a$. 
A logic $\mathcal{L}$ is called \emph{atomic}
if for every $b\in\mathcal{L}$ there exist an atom 
$a\in\mathcal{L}$ such that $a\le b$. 
A logic $\mathcal{L}$ is called \emph{atomistic}
whenever every $b\in\mathcal{L}$ can be written
as supremum of all atoms that are less or equal to $b$.
\end{dfn}

\begin{dfn}[see Def.~2.1.1 in \cite{ptak1991orthomodular}]
    Let $\mathcal L$ be a quantum logic.
    Function $\mu\colon\mathcal L \to [0,1]$ is called
    a \emph{probability measure} or \emph{state} on $\mathcal L$
    if and only if
    \begin{enumerate}[(i)]
        \item $\mu(0) = 0, \mu(\bbone) = 1$,
        \item $\mu(a_1\vee a_2\vee\dots) = \sum_{k=1}^\infty \mu(a_k)$,
            whenever $a_i \le a_j\cmpl$ for $i\neq j$.
    \end{enumerate}
    \label{thm:state}
\end{dfn}

\begin{dfn}[see Def.~4.1.1 in \cite{ptak1991orthomodular}]
    \label{thm:observables} Let $\mathcal L$ be a quantum logic. Function
    $X\colon \borel(\reals)\to\mathcal L$, where $\borel(\reals)$ is a family
    of Borel sets on $\reals$, is called an \emph{$\mathcal L$-valued
        measure} on $B(\reals)$ or an \emph{observable} if
    \begin{enumerate}[(i)]
        \item $X(\reals) = 1$,
        \item $X(\reals\setminus A) = X(A)\cmpl$,
        \item $X(A_1\cup A_2\dots) = X(A_1)\vee X(A_2)\vee \dots$ for any
            countable family $\set{A_i}\subset\borel(\reals)$ of pairwise
            disjoint sets.
    \end{enumerate}
\end{dfn}

\begin{example}
    Let $\hilbert H$ be a separable Hilbert space. Denote by $\mathcal
    P(\hilbert H)$ the set of orthogonal projections on $\hilbert H$, ordered
    by the subspace inclusion. Define $P\cmpl = \bbone - P$, where $\bbone$ is the
    identity operator. Then $\mathcal P(\hilbert H)$ is a quantum logic.
    Observe that for any pair $P, Q\in\mathcal P(\hilbert H)$, $P\vee Q$ and
    $P\wedge Q$ exist. A quantum logic with such property is called an
    \emph{orthomodular lattice}.

    $\mathcal L$-valued measures correspond via the
    spectral theorem to the self-adjoint operators acting
    on the Hilbert space $\hilbert H$.

    Similarly, probability measure on $\mathcal P(\hilbert H)$
    defines a quantum mechanical state described by a density matrix.
    When the measure is supported by only one rank-1 projector,
    then the corresponding state is pure.
\end{example}

\begin{example}
    Let $\Gamma$ be a phase space of some classical model
    with some measure $\nu$.
    Denote by $\mathcal B(\Gamma)$ a family
    of $\nu$-measurable subsets of $\Gamma$,
    ordered by set inclusion.
    For any $A\in\mathcal B(\Gamma)$ we set
    $A\cmpl = \Gamma\setminus A$.
    Then $\mathcal B(\Gamma)$ is a quantum logic.
    Since $A\vee B = A\cup B$ and $A\wedge B = A\cap B$,
    it is also an orthomodular lattice.
    Moreover, for any triple $A, B, C\in\mathcal B(\Gamma)$
    distributivity law is satisfied:
    \begin{equation*}
        A\vee(B\wedge C) = (A\vee B)\wedge(A\vee C).
    \end{equation*}
    Orthomodular lattice for which the distributivity law holds
    is called a \emph{Boolean algebra}.
\end{example}

Strong connection between the structure of Boolean algebra and the Kolmogorov's
axiomatization of probability motivates the use of an orthomodular lattice or
an orthomodular poset as an axiomatic definition of generalized probability
\cite{pulmanova2007}.
Elements of a quantum logic $\mathcal L$
are interpreted as \emph{events}.
In physical terms, these correspond to two-valued measurements, called \emph{propositions}.
Outcomes of a proposition can be conveniently labeled as ``yes'' and ``no''.
For a given state $\mu$ on $\mathcal L$ and $a\in\mathcal L$
value $\mu(a)$ is interpreted as a probability of getting answer ``yes''
for proposition $a$.
If $a\le b\cmpl$ we say that $a$ and $b$ are \emph{disjoint}.

Quantum mechanics taught us that it is important to distinguish subsets of
propositions that can be described by a classical probability model. This
motivates the following definition:

\begin{dfn}[see Def. 1.2.1 in \cite{ptak1991orthomodular}]
    Let $\mathcal L$ be a quantum logic.
    Elements $a, b\in\mathcal L$ are said to be \emph{compatible},
    what will be denoted by $a \leftrightarrow b$,
    whenever there exist mutually disjoint propositions $a_1, b_1, c$
    such that: $a = a_1 \vee c$, $b = b_1 \vee c$.
\end{dfn}

\begin{dfn}[see Def. 1.2.2 in \cite{ptak1991orthomodular}]
    A subset $\mathcal K \subset \mathcal L$ of quantum logic $\mathcal L$
    is called a \emph{sublogic} of $\mathcal L$ whenever:
    \begin{enumerate}[(i)]
        \item $0\in\mathcal K$,
        \item $a\in\mathcal K$ implies $a\cmpl\in\mathcal K$,
        \item for a countable family $\set{a_i}\subset K$ of
            pairwise disjoint elements,
            supremum $\bigvee \set{a_i}$ taken in $\mathcal L$
            belongs to $\mathcal K$.
    \end{enumerate}
    If $\mathcal K$ is a maximal Boolean sublogic of $\mathcal L$,
    then it is called a \emph{block}.
\end{dfn}

A compatible pair, $a, b$, of propositions from $\mathcal L$ can be
simultaneously measured, i.e.\ there exists a Boolean sublogic of $\mathcal L$
that contain $a$ and $b$ (Thm.~1.3.23 of Ref.~\cite{ptak1991orthomodular}). But
it is important to remember that, in general, if $A \subset \mathcal L$ is a
subset of orthomodular poset $\mathcal L$ such that all elements of $A$ are
pairwise compatible, then there might not exist any Boolean sublogic $\mathcal
B\subset\mathcal L$ such that $A\subset \mathcal B$. For that we need a
stronger definition of compatibility of sets. In this paper we will not
consider anything more than a compatible pair, so for sake of brevity we will
omit details (see Sec.~1.3 of Ref.~\cite{ptak1991orthomodular}).

Following class of quantum logics will be of profound importance for us:
\begin{dfn}[cf.\ Sec. 1.1 in \cite{ptak1991orthomodular}]
    Let $\Omega$ be a set and $\Delta$ be a family of subsets of $\Omega$ such
    that:
    \begin{enumerate}[(i)]
        \item $\emptyset\in\Delta$,
        \item if $A\in\Delta$ then $\Omega\setminus A\in\Delta$,
        \item for any countable family $\set{A_i}_{i\in I}\subset \Delta$ of
            pairwise disjoint sets $\bigcup_{i\in I} A_i \in \Delta$.
    \end{enumerate}
    We say that $(\Omega, \Delta)$ is a \emph{concrete logic} with partial order
    induced by set inclusion.
    \label{thm:concrete-logic}
\end{dfn}

\begin{example}
    Let $\Omega = \set{1,\dots,2n}$
    and $\Delta$ be a family of subsets of $\Omega$
    with even number of elements.
    Then $(\Omega, \Delta)$ is a concrete logic
    which is Boolean algebra for $n=1$,
    orthomodular lattice for $n=2$
    and quantum logic for $n\ge3$.
\end{example}

In order to determine whether a logic is concrete or not, we need few more
notions and results, namely:

\begin{dfn}[see Def.\ 44 in \cite{pulmanova2007}]
    Let $\mathcal L$ be a quantum logic.
    The set of states $\mathcal S$ is said to
    be \emph{rich} if for any disjoint pair of propositions $p,q$ there exists a
    state $\sigma\in\mathcal S$ such that $\sigma(p) = 1$ and $\sigma(q) > 0$.

    We say that $\mathcal L$ is \emph{rich} whenever it has a rich subset of states.
    We say that $\mathcal L$ is \emph{2-rich} whenever it has a rich set of two-valued
    states (i.e.\ states with a property that 
    $\forall q\in L, \sigma(q) = 1$ or $\sigma(q)=0$).
    \label{thm:rich}
\end{dfn}

\begin{thm}[see Thm.\ 48 in \cite{pulmanova2007}]
    \label{thm:set-repr}
    A quantum logic $\mathcal L$ is set-representable,
    i.e.\ there exists order preserving isomorphism between $\mathcal L$ and some concrete
    logic $(\Omega, \Delta)$ if and only if $\mathcal L$ is 2-rich.
\end{thm}

Concrete logics exhibit some classical properties. Let us firstly summon a
reformulation of Heisenberg uncertainty relation in the language of quantum
logics. Following Ref.~\cite{pulmanova2007}, let $X$ be a real observable on
the quantum logic $\mathcal L$. Then for any state $\mu$ on $\mathcal L$ we
define an expected value of $X$ as
\begin{equation*}
    \mu(X) \defeq \int_\reals t \mu(X(dt)),
\end{equation*}
whenever the integral exists.
Similarly we define a variance of $X$ in state $\mu$
\begin{equation*}
    \Delta_\mu X \defeq \int_\reals (t-\mu(X))^2 \mu(X(dt)).
\end{equation*}
Then for a given pair of real observables $X, Y$ either
\begin{equation}
\fl    \forall \varepsilon>0\,
    \exists \hbox{ a state }\mu \hbox{ with finite variance for } X \hbox{ and } Y,\qquad
    (\Delta_\mu X)(\Delta_\mu Y) < \varepsilon
    \label{eq:heisenber-fail}
\end{equation}
or
\begin{equation}
\fl    \exists \varepsilon>0\,
    \forall \hbox{ states } \mu \hbox{ with finite variance for } X \hbox{ and } Y,\qquad
    (\Delta_\mu X)(\Delta_\mu Y) \ge \varepsilon.
    \label{heisenberg}
\end{equation}
In the former case we say that \emph{Heisenberg uncertainty relations are not
satisfied} while in the latter one we say that \emph{Heisenberg uncertainty
relations are satisfied}. We have:

\begin{thm}[cf.\ Thm.~50 and Thm.~129 in \cite{pulmanova2007}]
    If $\mathcal L$ is a concrete logic,
    then Heisenberg uncertainty relations are not satisfied.
    \label{thm:heisenberg}
\end{thm}

\section{Propositional system of the Box World}
\label{sec:logics-of-bwo}

In this section we construct the propositional system of $(2, 2)$-box world
described in the Introduction (cf.\
Ref.~\cite{barrett2005nonlocal,short2010strong}). To avoid confusion with
notions introduced in the previous section, a $(2, 2)$-box world state in the
sequel will be called \emph{PR-box state} and the notion of \emph{state} will
be reserved for a measure on an orthomodular poset. The construction is a
rather standard procedure, involving arguments similar to the ones used by
Mackey in his axiomatic approach to quantum mechanics \cite{mackey1963mathematical}
(we remark that the framework we use is more general than either
Mackey's axioms or Piron's axioms, \cite{piron1976foundations}). We emphasize
that the procedure does not add anything new to the original definition of
the (2,2)-box world. We only explore the structure that it already has.

We start by making few observations. Firstly, the logic $\mathcal L$ of
the (2,2)-box world, if it exists, must contain propositions corresponding to the
most elementary questions in $(2, 2)$-box world, i.e.\ questions of the form
``does measuring $a$ on the first subsystem and $b$ on the second yields the
result $\alpha$ on the first subsystem and $\beta$ on the second''. We will
denote the corresponding proposition by $[ab, \alpha \beta]$.

Secondly, any PR-box state on (2,2)-box world should correspond to some state
on $\mathcal L$. Consequently, if $P$ is a PR-box state then the state $\rho_P$
should satisfy
\begin{equation*}
    \rho_P\big([ab, \alpha \beta]\big) = P(\alpha \beta| ab).
\end{equation*}
By the definition and properties of the Box World theory
(cf.\ Thm~2 in Ref.~\cite{short2010strong}),
propositions from the set
\begin{equation}
	\label{eq:atoms}
        \mathcal A = \set{[ab,\alpha \beta]\setdef a,b\in\set{x,y}, \alpha,\beta\in\set{0,1}}
\end{equation}
are sufficient to describe completely any measurement in (2,2)-box world, so
any other propositions must be build from the elements of $\mathcal A$.
Moreover, $\mathcal L$ must contain two trivial propositions, which we denote
by $\bbone$ (trivial ``yes'') and $0$ (trivial ``no'').

Given a question corresponding to the proposition $[ab, \alpha \beta]$ we can
always ask the negated question 
(i.e.\ interchange answer ``yes'' with ``no''). 
We will denote
the proposition described by such a question by $[ab, \alpha \beta]\cmpl$.
Because we expect that the set of the PR-box states is rich enough to
distinguish different observables and determine their ordering, we can formally
define $[ab, \alpha \beta]\cmpl$ by
\begin{equation}
    \rho_P([ab,\alpha\beta]\cmpl) = 1 - \rho_P([ab,\alpha\beta])\qquad
    \forall \hbox{ PR-box states } P,
\end{equation}
and $r \le q$ if and only if
\begin{equation}
    \rho_P(r) \le \rho_P(q),\qquad\forall \hbox{ PR-box states } \rho.
\end{equation}
If for some fixed pair of propositions $r \le q\cmpl$,
then
\begin{equation}
    \rho_P(p) + \rho_P(q) \le 1,\qquad
    \forall \hbox{ PR-box states } P,
    \label{eq:sum}
\end{equation}
and so $r$ and $q$ cannot both be true.
Thus in principle, the question: ``does $r$ or $q$ is true?''
should make sense.
We denote the proposition corresponding to such question
by $r\oplus q$ and define it by
\begin{equation}
	\rho_P(p\oplus q) = \rho_P(p) + \rho_P(q).
\end{equation}
We could proceed without this assumption,
but from the operational point of view it is well justified
and this seems to be in line with the general idea of non-signaling theories.
Moreover, this construction is essential to identify propositions about a single subsystem
and can be used to express complementary propositions, e.g.
\begin{equation*}
    [xx, 00]\cmpl = [xx,01] \oplus [xx,10] \oplus [xx,11].
\end{equation*}
By construction, the operation $\oplus$ is commutative and associative.

In order to generate all possible propositions in (2,2)-box world, we define
recursively a sequence of sets:
\begin{eqnarray*}
    \mathcal L_0 &=\mathcal A\cup \set{0,\bbone}, \\
    \mathcal L_{i+1} &= \mathcal L_i\cup\mathcal L_{i+1}^c\cup \mathcal L_{i+1}^p
\end{eqnarray*}
where,
\begin{eqnarray*}
\mathcal L_{i+1}^c &= \set{q\cmpl; q\in\mathcal L_i}, \\
\mathcal L_{i+1}^p &= \set{r\oplus q; r,q\in\mathcal L_i,\hbox{ s.t. } \forall
      \hbox{ PR-box states }P\,\,\, \rho_P(r) + \rho_P(q)\leq 1}.
\end{eqnarray*}
Computing $\mathcal L_i^p$ requires maximizing the left hand side of
Eq.~(\ref{eq:sum}) with respect to PR-box states. If we represent a PR-box
state $P$ as a matrix (table):
\begin{equation}
	\rho_P = (\rho_{ij}) = \left(
	\begin{array}{cccc}
		P(00|xx) & P(00|xy) & P(00|yx) & P(00|yy) \\
		P(01|xx) & P(01|xy) & P(01|yx) & P(01|yy) \\
		P(10|xx) & P(10|xy) & P(10|yx) & P(10|yy) \\
		P(11|xx) & P(11|xy) & P(11|yx) & P(11|yy) \\
	\end{array}
	\right), \label{eq:state-m}
\end{equation}
then the properties P1, P2, P3 are just linear constraints on $\rho_{ij}$. The
non-signaling conditions become,
\begin{eqnarray}
	\label{eq:no-signal}
		\rho_{11}+\rho_{31}&=&\rho_{13}+\rho_{33}, \qquad \rho_{21}+\rho_{41}=\rho_{23}+\rho_{43},\nonumber \\
		\rho_{12}+\rho_{32}&=&\rho_{14}+\rho_{34}, \qquad \rho_{22}+\rho_{42}=\rho_{24}+\rho_{44},\nonumber \\
		\rho_{11}+\rho_{21}&=&\rho_{12}+\rho_{22}, \qquad \rho_{31}+\rho_{41}=\rho_{32}+\rho_{42},\nonumber \\
		\rho_{13}+\rho_{23}&=&\rho_{14}+\rho_{24}, \qquad \rho_{33}+\rho_{43}=\rho_{34}+\rho_{44},
\end{eqnarray}
The normalization conditions are given by
\begin{eqnarray}
	\label{eq:normalization}
		\rho_{11}+\rho_{21}+\rho_{31}+\rho_{41}&=&1,\nonumber \\
		\rho_{12}+\rho_{22}+\rho_{32}+\rho_{42}&=&1,\nonumber \\
		\rho_{13}+\rho_{23}+\rho_{33}+\rho_{43}&=&1,\nonumber \\
		\rho_{14}+\rho_{24}+\rho_{34}+\rho_{44}&=&1,
\end{eqnarray}
and positivity means that $\rho_{ij}\ge0$. Thus a test whether $r\oplus q$
exists is a linear programming problem that can be solved exactly. We used
Wolfram Mathematica \cite{mathdoc} for this purpose. Similarly, determining
whether $r\le q$ is also a linear programming problem, but now the objective
function is given by
\begin{equation*}
    \rho_P(r) - \rho_P(q) \le 0,
\end{equation*}
(in fact, $r\oplus q$ is defined whenever $r \le q\cmpl$,
so the former problem is a special case of the latter).
Let us also observe that due to the non-signaling constraints,
for some $r, q$ we have
\begin{equation*}
    \rho_P(r) = \rho_P(q),
\end{equation*}
i.e.\ $r \leq q$ and $q \leq r$.
We will write in that case that $r \sim q$
(e.g.\ $[xx, 00]\oplus [xx, 01] \sim [xy, 00]\oplus[xy, 01]$)

Finally, let us remark that for (2,2)-box world $\mathcal L_i = \mathcal
L_{i+1}$ for $i\ge 4$. It follows from the observation that subset of $\mathcal
A$ with pairwise orthogonal elements can have cardinality at most $4$. Then we
define $\mathcal L$ as a quotient of $\mathcal L_4$ with respect to the
equivalence relation $\sim$, i.e.\ $\mathcal L = \mathcal L_4/\sim$

The partially ordered set $\mathcal L$ has 82 elements
(cf.\ Table~\ref{tab:L} for equivalence class representatives).
Any further property of $\mathcal L$ can be obtained
by analyzing the directed graph representing partial ordering of $\mathcal L$
(see Figure~2 for its schematic representation).
Again, we used Mathematica to obtain all the results reported in the sequel.
We want to emphasize, that the results are exact,
as they follow from traversing the mentioned graph.
We sum up basic properties of $\mathcal L$ in the
following proposition:

\begin{thm}
    Let $\mathcal L$ be an above constructed
    set of propositions about $(2, 2)$-box world.
    Then
    \begin{enumerate}
        \item[(i)] $\mathcal L$ is an atomistic quantum logic,
        \item[(ii)] let
            \begin{eqnarray*}
                x_\alpha\bbone &=& [xx,\alpha0]\oplus[xx,\alpha1],
                \bbone x_\alpha = [xx,0 \alpha]\oplus[xx,1 \alpha],\\
                y_\alpha\bbone &=& [yx,\alpha0]\oplus[yx,\alpha1],
                \bbone y_\alpha = [xy,0 \alpha]\oplus[xy,1 \alpha],
            \end{eqnarray*}
            then all the pairs of the propositions
            $(x_\alpha \bbone, \bbone x_\beta)$,
            $(x_\alpha \bbone, \bbone y_\beta)$,
            $(y_\alpha \bbone, \bbone x_\beta)$,
            $(y_\alpha \bbone, \bbone y_\beta)$
            are compatible while pairs
            $(x_\alpha \bbone, y_\beta \bbone)$,
            $(\bbone x_\alpha, \bbone y_\beta)$
            are not,
        \item[(iii)] the set of all states on $\mathcal L$
            coincide with the set of all PR-box states on $\mathcal L$,
        \item[(iv)] the logic $\mathcal L$ is set-representable,
        \item[(v)] there are pairs of blocks of $\mathcal L$
            that have two atoms in common,
            (i.e.\ blocks do not form an almost disjoint system,
            see Def. 2.4.2 in \cite{ptak1991orthomodular}).
    \end{enumerate}
    \label{thm:main-struct}
\end{thm}
\textbf{Proof}
	To show (i) we directly check requirements of Def.~\ref{thm:qlogic} (we
    perform the exhaustive check \cite{gitrepo}). It is interesting to note,
    that $\mathcal L$ is not a lattice, as 32 pairs of propositions do not have
    unique least upper bound. For example, the minimal elements of the upper
    bound of $[xx, 00]$ and $[yy, 00]$ are $[xy, 11]\cmpl$ and $[yx, 11]\cmpl$.

    Similarly, (ii) can also be checked directly using the definition.
    For example, for $x_1 \bbone$, $\bbone x_1$ and
    \begin{equation*}
        a = [xx, 10], \quad b = [xx, 01], \quad c = [xx, 11],
    \end{equation*}
    we have that $x_1\bbone = a \vee c, \bbone x_1 = b \vee c$, and all $a, b,
    c$ are mutually disjoint. To observe non-compatibility of the remaining
    pairs, e.g.\ $(x_1\bbone, y_1\bbone)$ it is enough to observe that
    $(x_1\bbone)\cmpl = x_0\bbone$ and
    \begin{equation*}
        x_1\bbone \vee (x_0\bbone\wedge y_1\bbone) \neq
        (x_1\bbone \vee x_0\bbone)\wedge(x_1\bbone \vee y_1\bbone),
    \end{equation*}
    thus the distributivity law does not hold and, consequently,
    $x_1\bbone, y_1\bbone$ do not span a Boolean algebra.

    To show (iii) we use the fact that $\mathcal L$ is, by
    construction, an atomistic logic and any state is determined by
    the value on its atoms. The set of atoms of $\mathcal L$ is
    exactly the set $\mathcal A$, so we have obvious mapping between
    states on $\mathcal L$ and PR-box states on the $(2, 2)$-box world. 
    We need to show that the restrictions imposed on states by the order
    structure on $\mathcal L$ (see Def.~\ref{thm:state}) are not
    weaker than P1--P3. Again, by an exhaustive check we observe that
    the former and the latter linearly dependent on each other. Note,
    that this is not a trivial property, because as we observe in
    Sec.~\ref{sec:restr-class-case}, we obtain the same order
    structure on $\mathcal L$ with much restricted set of PR-box
    states. 
   
    For (iv) we show again by the exhaustive check of all $2^{16}$
    matrices with binary entries that the set of states on $\mathcal L$
    is 2-rich (there are exactly 16 two-valued states which are, in fact,
    extreme states of the so-called classically correlated boxes,
    discussed in more details in Sec.~\ref{sec:restr-class-case}).
    Consequently, by Thm.~\ref{thm:set-repr}, $\mathcal L$ is set-representable.
    Concrete representation will be given in Sec.~\ref{sec:restr-class-case}.

    Finally, for (v) we can easily compute all blocks of $\mathcal L$.
    They are spanned by the following subsets of atoms of $\mathcal L$:
    \begin{eqnarray*}
        \fl \set{[xx, 00], [xx, 01], [xx, 10], [xx, 11]}, \quad
            \set{[xy, 00], [xy, 01], [xx, 10], [xx, 11]},\\
        \fl \set{[xx, 00], [xx, 01], [xy, 10], [xy, 11]}, \quad
            \set{[xy, 00], [xy, 01], [xy, 10], [xy, 11]},\\
        \fl \set{[yx, 00], [xx, 01], [yx, 10], [xx, 11]}, \quad
            \set{[xx, 00], [yx, 01], [yx, 11], [xx, 10]}, \\
        \fl \set{[yx, 00], [yx, 01], [yx, 10], [yx, 11]}, \quad
            \set{[yy, 00], [yy, 01], [yx, 10], [yx, 11]}, \\
        \fl \set{[yy, 00], [xy, 01], [yy, 10], [xy, 11]}, \quad
            \set{[xy, 00], [yy, 01], [yy, 11], [xy, 10]}, \\
        \fl \set{[yx, 00], [yx, 01], [yy, 10], [yy, 11]}, \quad
            \set{[yy, 00], [yy, 01], [yy, 10], [yy, 11]}.
    \end{eqnarray*}
    This last property has an important consequence:
    quantum logic $\mathcal L$ of (2,2)-box world cannot be
    represented by a Greechie diagram
    (cf.\ \cite{ptak1991orthomodular}, Sec. 2.4).
$\blacksquare$
\medskip

Thus we see that it is actually justified to call a function $P$ from the
definition of the $(2,2)$-box world \emph{a probability}, but we should expect that
it will have some properties that are not possessed by quantum probability.
Clearly, Eq.~(\ref{eq:chsh}) express one of them.

Moreover, propositions that we can identify with propositions about the left or
right box, i.e.\ $x_\alpha \bbone, \bbone x_\alpha$, etc., are compatible, thus
simultaneously measurable what is required by the principle of relativistic
causality. Note that it seems that compatibility is, in general, more
restrictive than non-signaling, thus the problem of compatibility in
non-signaling boxes with higher number of inputs and outputs should be
addressed in the future.

An interesting property (iii) will be further discussed in the light of results
presented in Sec.~\ref{sec:restr-class-case}. 
Moreover, it allows us to drop the
distinction between states and PR-box states in the sequel.

\begin{figure}[h]
\label{fig:2-2-box-lattice}
    \includegraphics[width=15cm,height=5cm]{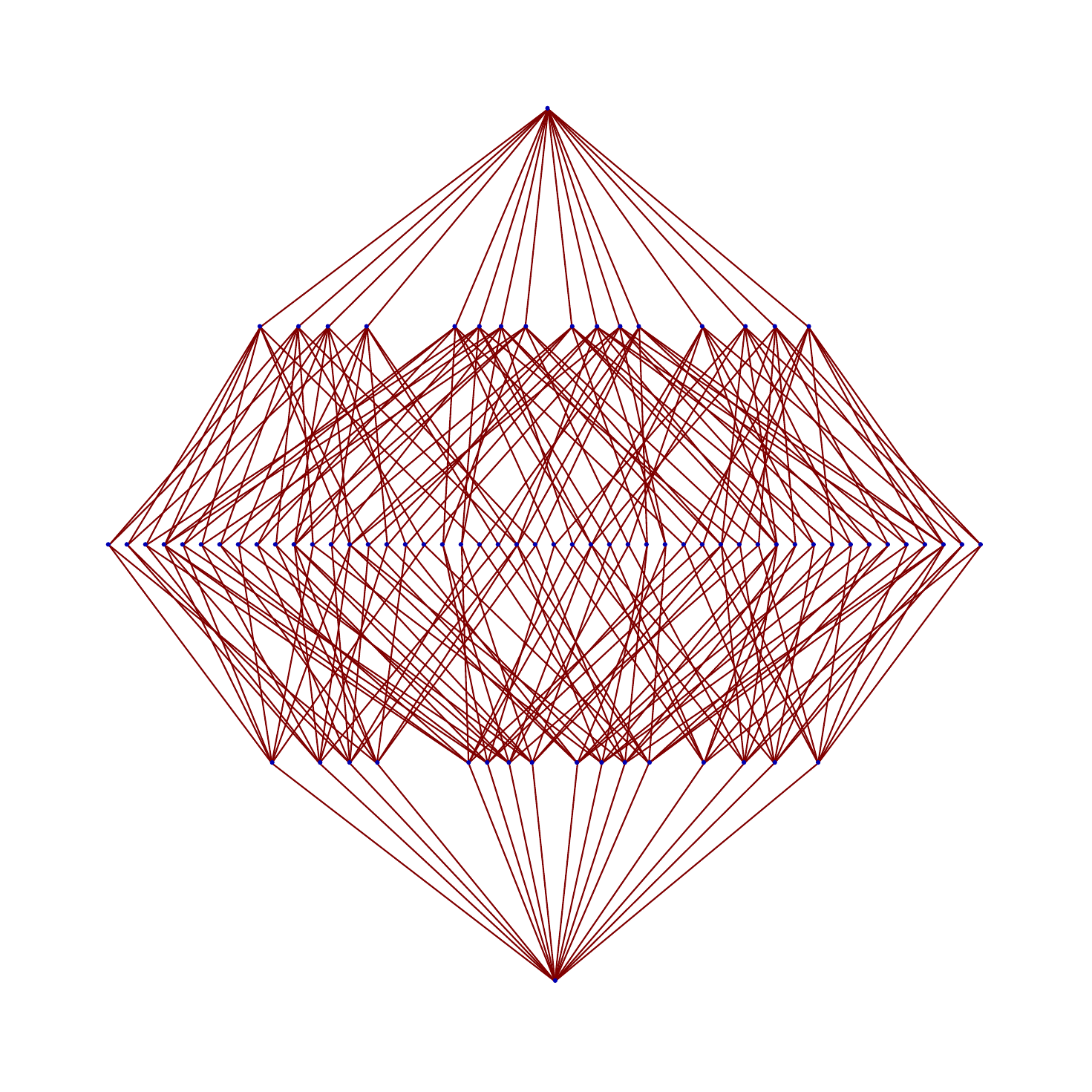}
    \caption{Hasse diagram of (2,2)-box world logic.
        Labels are omitted for compactness.
        There is no corresponding Greechie diagram, as blocks
        of (2,2)-box world logic are not almost disjoint.}
\end{figure}

\section{Properties of the (2,2)-box world}

In the previous section we constructed logic of the (2,2)-box world
and showed that it has a structure that is consistent with
na\"ive interpretation of non-signalling theories.
Here we will explore more deeply its properties.

\subsection{(2,2)-box world vs.\ two qubit system} \label{sec:box-world-one}

The most profound difference between the (2,2)-box world and
two qubit system stems from the fact that the former,
being a concrete logic, does not satisfy Heisenberg uncertainty relations
(see Thm.\ \ref{thm:heisenberg}).
So from that perspective, (2,2)-box world is more classical than two qubit system.

When it comes to the violation of Tsirelson bound, in the light of results
summarized in Thm.~\ref{thm:main-struct} it is clear that we can expect
qualitative difference, when the orthomodular lattice structure (the logic of
quantum models) is replaced by the more general structure of quantum logic. 
In the similar way the classical bound of 2 for CHSH-type inequalities is violated when
the Boolean algebra is replaced by a more general orthomodular lattice
structure.

Another essential difference between the (2,2)-box world and the two qubit
system lies in the fact that the logic of even the simplest quantum model has
infinite number of propositions, while (2,2)-box world has only 82 of them. In
particular, we cannot claim that the quantum mechanics is a special case of the
theory of non-signaling boxes. To see how the number of propositions is
important for the logical structure let us consider the following example.

Consider the two qubit system, described by $\hilbert H =
\complexes^2\otimes\complexes^2$. Let $P\otimes\bbone, Q\otimes\bbone\in
\bmaps(\hilbert H)$ be a pair of non-commuting projectors.
Analogously, choose two non-commuting projectors $\bbone\otimes R,
\bbone \otimes S$ representing measurement on the second subsystem. Now the join
of $P\otimes S$ and $Q\otimes R$ exists and by definition is equal to the projector
onto the smallest subspace that contains both images of $P\otimes S$ and
$Q\otimes R$. But this projector can not be expressed by a linear combination of
projectors that we used so far.

In other words, if we denote projector $P\otimes S$ by $[xy,11]$,
$(\bbone-P)\otimes S$ by $[xy,01]$, $Q\otimes R$ by $[yx,11]$, $Q\otimes
(\bbone-R)$ by $[yx,10]$, etc. then $[xy,11]\vee [yx,11]$ in quantum mechanics
exists, but it is not a linear combination of $[ab,\alpha\beta]$. On the other
hand, in the (2,2)-box world there is no equivalent of the above quantum mechanical
proposition. Consequently, the more general structure of the (2,2)-box world is a
result of a depleted set of propositions when compared to quantum mechanics. We
elaborate on this statement in the Sec.~\ref{sec:embedding-box-world}.

To sum up, despite the obvious fact that a quantum logic is a more general
object than an orthomodular lattice, simple statement that (2,2)-box world is a
generalization of some quantum model (e.g.\ two qubits) is not justified.

\subsection{Restriction to the classical case} \label{sec:restr-class-case}

We would like to analyze the logical structure of (2,2)-box world when we
restrict to states which are convex combinations of the following 16 PR-box
states \cite{barrett2005nonlocal}:
\begin{equation}
\fl        P_{mnlk}(\alpha \beta|ab) =
	\left\{ \begin{array}{ll}
		1 & \mathrm{if\ } \alpha = m a + n \mathrm{\ mod\ 2,\ } \beta = l b + k \mathrm{\ mod\ }2,\\
		0 & \mathrm{otherwise},
	\end{array}\right.
        \label{eq:ccboxes}
\end{equation}
where $m,l,n,k\in\set{0, 1}$ and intput values $x, y$ are treated as $0, 1$ respectively.
These states are
called \emph{classically correlated boxes} or \emph{local boxes}
\cite{barrett2005nonlocal}. We repeat the construction of the logic from the
Sec.~\ref{sec:logics-of-bwo} but with respect to this restricted set of states.
It is quite amazing that we get exactly the same logic $\mathcal L$.
Consequently, these states are not ``classical'' in the sense of forming a
closed set to which we can apply classical probability rules, but indeed only
``classically correlated'', i.e.\ chosen in such way that they do not violate
CHSH-like inequalities. Let us thus examine the relation of the ``classically
correlated'' boxes to the truly classical boxes, i.e.\ boxes implemented by
classical system.

In classical physics a system is described by its phase space. For system
that can be fully described by two dichotomic observables the phase space is a
set of four points. The phase space of a compound system is a
Cartesian product of phase spaces of the components, thus the phase space of
truly the classical (2,2)-box world is a 16 element set. We can think of it as
\begin{equation*}
	\Gamma = \set{(a,b,c,d)\setdef a,b,c,d\in\set{0,1}},
\end{equation*}
where each 4-tuple represents the values of four dichotomic observables in that
point of $\Gamma$ (i.e.\ a pure state). We assume that $a,b$ are the values of
$x$ and $y$, respectively on the first subsystem and $c,d$ are values of $x$
and $y$ on the second subsystem.

Any probability measure on this set can be represented by the point within the
15-simplex:
\begin{equation*}
	(p_1, p_2, \dots, p_{15}),\quad \hbox{ where } p_i\ge0, \sum_i p_i \leq 1.
\end{equation*}
It is important to note that any point of a simplex can be uniquely represented
as a convex combination of extreme points. In other words a mixed state
``remembers'' how it was made. This constitutes a remarkable property of
classical theories.

To any proposition of type $[xy,\alpha \beta]$ there corresponds a subset
$E([xy,\alpha \beta])$ of $\Gamma$ defined in the following way:
\begin{equation}
\fl    E([z_1 z_2,\alpha \beta]) \defeq
	\left\{
    \begin{array}{ll}
        \set{(\alpha, b, \beta, d)\in \Gamma\setdef b,d \in\set{0,1}}& \mathrm{ for\ } z_1 = x, z_2 = x\\
        \set{(\alpha, b, c, \beta)\in \Gamma\setdef b,c \in\set{0,1}}& \mathrm{ for\ } z_1 = x, z_2 = y\\
        \set{(a, \alpha, \beta, d)\in \Gamma\setdef a,d \in\set{0,1}}& \mathrm{ for\ } z_1 = y, z_2 = x\\
        \set{(a, \alpha, c, \beta)\in \Gamma\setdef a,c \in\set{0,1}}& \mathrm{ for\ } z_1 = y, z_2 = y
    \end{array}
    \right.
    \label{eq:quest-phase}
\end{equation}
We can define a mapping $\varphi$ from the simplex of classical states on
$\Gamma$ into the set of classically correlated boxes on (2,2)-box world. Using
the matrix representation of (2,2)-box world states (\ref{eq:state-m}) one
expresses it explicitly
\begin{equation}
\fl	\varphi(\mu) = \rho_\mu = \left(
	\begin{array}{cccc}
		\mu(E[xx,00]) & \mu(E[xy,00]) & \mu(E[yx,00]) & \mu(E[yy,00]) \\
		\mu(E[xx,01]) & \mu(E[xy,01]) & \mu(E[yx,01]) & \mu(E[yy,01]) \\
		\mu(E[xx,10]) & \mu(E[xy,10]) & \mu(E[yx,10]) & \mu(E[yy,10]) \\
		\mu(E[xx,11]) & \mu(E[xy,11]) & \mu(E[yx,11]) & \mu(E[yy,11]) \\
	\end{array}
	\right),
  \label{eq:cl2bwo}
\end{equation}
where $\mu$ is a probability measure on $\Gamma$. It is easy to see that this map maps
pure classical states onto extreme classically correlated states \eref{eq:ccboxes} and
thus onto the set of all classically correlated states. But $\varphi$ is not injective.
As an example consider the following two probability measures:
\begin{eqnarray*}
  \mu_1 (\set{u}) &=& \left\{\begin{array}{ll}
    1/2 & \mathrm{if\ } u \in \set{(1, 0, 1, 1), (1, 1, 1, 0)}\\
    0   & \mathrm{otherwise}
  \end{array}\right. \\
  \mu_2 (\set{u}) &=& \left\{ \begin{array}{ll}
    1/4 & \mathrm{if\ } u \in \set{(1, 0, 1, 1), (1, 1, 1, 0), (1, 1, 0, 0), (1, 1,
    0, 1)}\\
    0   & \mathrm{otherwise}
  \end{array}\right.
\end{eqnarray*}
Image of both of them under the mapping \eref{eq:cl2bwo} equals
\begin{equation*}
  \varphi(\mu_1) = \varphi(\mu_2) = \left(\begin{array}{llll}
    0 & 0 & 0 & 0\\
    0 & 0 & 0 & 0\\
    \frac12 & \frac12 & \frac12 & \frac12 \\
    \frac12 & \frac12 & \frac12 & \frac12
  \end{array}\right)
\end{equation*}
Consequently this is not an affine isomorphism and the set of classically
correlated states is not ``classically shaped'', as the image of $\varphi$ is
not a simplex. This means that classically correlated states do not decompose
in a unique way into a convex combination of extremal states. As a result, the
measurement in the (2,2)-box world must be destructive, even when we restrict
PR-box states to classically correlated states. Let us consider this in more
detail.

Assume that we are given ``sources'' of different (2,2)-box world boxes and a
device that allows mixing them. Denote by
\begin{eqnarray*}
  \rho_1 &=& \left(\begin{array}{llll}
    0 & 0 & 0 & 0\\
    0 & 0 & 0 & 0\\
    1 & 1 & 1 & 1\\
    0 & 0 & 0 & 0
  \end{array}\right) \quad
  \rho_2 = \left(\begin{array}{llll}
    0 & 0 & 0 & 0\\
    0 & 0 & 0 & 0\\
    0 & 0 & 0 & 0\\
    1 & 1 & 1 & 1
  \end{array}\right)\\
  \rho_3 &=& \left(\begin{array}{llll}
    0 & 0 & 0 & 0\\
    0 & 0 & 0 & 0\\
    1 & 0 & 1 & 0\\
    0 & 1 & 0 & 1
  \end{array}\right) \quad
  \rho_4 = \left(\begin{array}{llll}
    0 & 0 & 0 & 0\\
    0 & 0 & 0 & 0\\
    0 & 1 & 0 & 1\\
    1 & 0 & 1 & 0
  \end{array}\right)
\end{eqnarray*}
Now we set up our mixing device to output $\rho_1$ and $\rho_2$ with
probability $1/2$. The resulting state encoding our uncertainty would be $\rho
= 1/2\rho_1 + 1/2\rho_2$. Now we ask $[xx,11]$ question. If the answer is
``yes'', then on the same boxes we ask $[xy,11]$. Because we mix states
$\rho_1$ and $\rho_2$, in classical world we should get answer ``yes''. But in
(2,2)-box world we can obtain the same state $\rho$ by mixing completely
different states: $1/2\rho_3+1/2\rho_4$. In the latter case, for the second
question we should always get answer ``no''. To overcome this ambiguity we must
assume that either each box can be measured only once and then is destroyed, or
that the box undergoes state transformation under measurement. But then we need
to \emph{postulate} how the state is changed (e.g.\ after the positive answer to
question $[xy,ab]$ it transforms to uniform mixture of all extreme states in
which $[xy,ab]$ is certain).

Finally, the correspondence described in \eref{eq:quest-phase} between questions in
the (2,2)-box world and subsets of $\Gamma$ is an order preserving isomorphism between the
logic of the (2,2)-box world and subsets of $\Gamma$. Thus $(\Gamma, E(\mathcal L))$ is a
concrete logic corresponding to $\mathcal L$. Although this corollary seems to be
obvious, we verified it by ``by Mathematica'' independently.

\subsection{Embedding the (2,2)-box world into orthomodular lattice}
\label{sec:embedding-box-world}

The map $\varphi$ defined in the previous section induces an embedding
\begin{equation*}
  \mathrm{\fbox{\parbox{10em}{classically correlated\\
  (2,2)-box world boxes}}} \hookrightarrow
  \mathrm{\fbox{\parbox{8em}{classical model\\
  \scriptsize(Boolean algebra on $\Gamma$)}}}
\end{equation*}
which is understood in the following sense: propositional system is embedded
into larger structure and states are extended to states on this larger
structure (the map is not injective, but we can take an arbitrary element in
the preimage of $\varphi$ for each classically correlated (2,2)-box world
state). Similar construction can be performed for ``quantumly correlated
states''. Using the analogy discussed in Sec.~\ref{sec:box-world-one} it is
clear that we can define a map $\psi$
\begin{equation*}
\fl    \psi(\sigma) =\left(
    \begin{array}{llll}
        \Tr(\sigma P' \otimes R')&
        \Tr(\sigma P' \otimes S')&
        \Tr(\sigma Q' \otimes R')&
        \Tr(\sigma Q' \otimes S')\\
        \Tr(\sigma P' \otimes R )&
        \Tr(\sigma P' \otimes S )&
        \Tr(\sigma Q' \otimes R )&
        \Tr(\sigma Q' \otimes S )\\
        \Tr(\sigma P  \otimes R')&
        \Tr(\sigma P  \otimes S')&
        \Tr(\sigma Q  \otimes R')&
        \Tr(\sigma Q  \otimes S')\\
        \Tr(\sigma P  \otimes R )&
        \Tr(\sigma P  \otimes S )&
        \Tr(\sigma Q  \otimes R )&
        \Tr(\sigma Q  \otimes S )\\
    \end{array}\right),
\end{equation*}
where $\sigma$ is a density matrix of the two-qubit system, and $P' = \bbone - P,$ etc.
Clearly $\psi$ maps quantum states into a subset of states of (2,2)-box world, which we
would call ``quantumly correlated'', in the same manner as $\varphi$ does for classical
states. Consequently we have also an embedding

\begin{equation*}
{\fbox{\parbox{10em}{quantumly correlated\\
  (2,2)-box world boxes}}} \hookrightarrow
  {\fbox{\parbox{8em}{quantum model\\
  \scriptsize(projection lattice)}}}
\end{equation*}

Now it is natural to ask if we can embed the whole (2,2)-box world structure in
some larger orthomodular lattice, i.e.\ we ask if the following embedding exists
\begin{equation*}
  \mathrm{\fbox{(2,2)-box world boxes}} \stackrel{?}{\hookrightarrow}
  \mathrm{\fbox{orthomodular lattice}}
\end{equation*}
Intuition suggests that the answer to this question is negative because
orthomodular lattice is a propositional system of quantum mechanics and due to
violation of Tsirelson bound this embedding should not be possible.
Nevertheless, Authors are not aware of any proof of Tsirelson bound that relies
solely on the structure of orthomodular lattice. Moreover, there are
orthomodular lattices which are not lattices of projections of some von Neumann
algebra, so in principle violation of Tsirelson bound for orthomodular lattice
is possible. As orthomodular lattices seem to posses nicer physical
interpretation than more general orthomodular posets (for details see
Ref.~\cite{piron1976foundations}), the question risen in this paragraph becomes
interesting.

\begin{thm}
  The logic of (2,2)-box world cannot be embedded into orthomodular lattice in a
  way that preserves all (2,2)-box world states.
\end{thm}
\textbf{Proof}
  We will show this by contradiction. Let
  \begin{equation*}
    q_1 = [xx,11],\qquad q_2 = [yy,11].
  \end{equation*}
  Minimal upper bound of $q_1$ and $q_2$ consist of elements
  (one can use the graph of partial order to track this):
  \begin{equation*}
    r_1 = [xy,00]\cmpl,\qquad r_2 = [yx,00]\cmpl.
  \end{equation*}
  Assume that the unique element $q_1\vee q_2$ exists.
  By the definition, for any state $\rho$
  \begin{eqnarray}
      \rho(q_1) &\leq &\rho(q_1\vee q_2) \leq \rho(r_1),\\
      \rho(q_2) &\leq &\rho(q_1\vee q_2) \leq \rho(r_2).
    \label{eq:uniqjoinrestr}
  \end{eqnarray}
  Consider the following states
  \begin{equation*}
    \rho_1 = \left(
    \begin{array}{cccc}
      \frac{1}{2} & \frac{1}{2} & \frac{1}{2} & 0 \\
      0 & 0 & 0 & \frac{1}{2} \\
      0 & 0 & 0 & \frac{1}{2} \\
      \frac{1}{2} & \frac{1}{2} & \frac{1}{2} & 0 \\
    \end{array}
    \right),
    \qquad \rho_0 = \left(
    \begin{array}{cccc}
      0 & 0 & 0 & \frac{1}{2} \\
      \frac{1}{2} & \frac{1}{2} & \frac{1}{2} & 0 \\
      \frac{1}{2} & \frac{1}{2} & \frac{1}{2} & 0 \\
      0 & 0 & 0 & \frac{1}{2} \\
    \end{array}
    \right),
  \end{equation*}
  and convex combination $\rho_\lambda = \lambda \rho_1 + (1-\lambda)
  \rho_0$. It follows that $\rho_1(q_1\vee q_2)=1/2$ and
  $1/2\leq \rho_0(q_1\vee q_2)\leq 1$.
  For $1/4\leq \lambda \leq 3/4$ the state $\rho_\lambda$ is a
  classically correlated state. For $\lambda=3/4$ it can be equivalently written
  as a convex combination of the following eight classically correlated states, each
  with weight equal to $1/8$:
  \begin{eqnarray*}
    \sigma_1 &=& \left(
    \begin{array}{cccc}
      0 & 0 & 0 & 0 \\
      0 & 0 & 0 & 0 \\
      0 & 0 & 0 & 0 \\
      1 & 1 & 1 & 1 \\
    \end{array}
    \right) \qquad \sigma_2 = \left(
    \begin{array}{cccc}
      0 & 0 & 0 & 0 \\
      0 & 0 & 0 & 0 \\
      0 & 1 & 0 & 1 \\
      1 & 0 & 1 & 0 \\
    \end{array}
    \right) \\
    \sigma_3 &=& \left(
    \begin{array}{cccc}
      0 & 0 & 0 & 0 \\
      0 & 0 & 1 & 1 \\
      0 & 0 & 0 & 0 \\
      1 & 1 & 0 & 0 \\
    \end{array}
    \right) \qquad \sigma_4 = \left(
    \begin{array}{cccc}
      1 & 0 & 1 & 0 \\
      0 & 1 & 0 & 1 \\
      0 & 0 & 0 & 0 \\
      0 & 0 & 0 & 0 \\
    \end{array}
    \right)\\
    \sigma_5 &=& \left(
    \begin{array}{cccc}
      1 & 1 & 0 & 0 \\
      0 & 0 & 0 & 0 \\
      0 & 0 & 1 & 1 \\
      0 & 0 & 0 & 0 \\
    \end{array}
    \right) \qquad \sigma_6 = \left(
    \begin{array}{cccc}
      0 & 1 & 0 & 0 \\
      1 & 0 & 0 & 0 \\
      0 & 0 & 0 & 1 \\
      0 & 0 & 1 & 0 \\
    \end{array}
    \right) \\
    \sigma_7 &=&\left(
    \begin{array}{cccc}
      1 & 1 & 1 & 1 \\
      0 & 0 & 0 & 0 \\
      0 & 0 & 0 & 0 \\
      0 & 0 & 0 & 0 \\
    \end{array}
    \right) \qquad \sigma_8 = \left(
    \begin{array}{cccc}
      0 & 0 & 1 & 0 \\
      0 & 0 & 0 & 1 \\
      1 & 0 & 0 & 0 \\
      0 & 1 & 0 & 0 \\
    \end{array}
    \right)
  \end{eqnarray*}
  We check that due to~\eref{eq:uniqjoinrestr} we have
  $\sigma_i(q_1\vee q_2) = 1$ for $i=1,2,3$ and
  $\sigma_i(q_1\vee q_2) = 0$ for the remaining 5 states.
  Thus $\rho_{3/4}(q_1\vee q_2) = 3/8$.
  On the other hand
  \begin{equation*}
    3/4 \rho_1(q_1\vee q_2) + 1/4 \rho_0(q_1\vee q_2) = 3/8,
  \end{equation*}
  what immediately implies that $\rho_0(q_1\vee q_2)=0$ in contradiction with
  the previously obtained bounds on the value of $\rho_0(q_1\vee q_2)$.
  Consequently it is not possible to define a unique join if we want to allow
  all (2,2)-box world states to extend to valid states on the larger structure.
$\blacksquare$

\section{Outlook}
\label{sec:conclusions}

The presented analysis shows how (2,2)-box world emerges. We start with the
classical logic over 16-element phase space. Then we carefully select 82
propositions in a way that will allow us to interpret resulting structure in
terms of a system composed of two subsystems. In particular, we need to
preserve compatibility between questions that we want to assign to different
components (in the sense that questions are compatible whenever they span a
Boolean algebra). As a result we obtain an orthomodular poset $\mathcal L$ (the
order is induced by the order of classical logic). Finally, we take all
possible probability measures as admissible states.
In this way, due to the link between probability and logic, 
we obtain a generalized probability theory. 
One of the features of this theory is violation
of Tsirelson bound of the quantum probability theory.

This perspective ``hides'' the non-signaling condition (P3) in the appropriate
selection of 82 propositions from classical logic. This can possibly help to
define non-signaling systems that consist of more than two boxes, as now our
main concern is compatibility of certain questions: the notion which has clear
meaning in any orthomodular poset.

Moreover, presented link between non-signaling theories and quantum logics can
shed new light on the problem of describing composite systems in the language
of quantum logics (as far as Authors know, there is no unique or canonical way
to build the logic of composite systems from logics of components,
cf.~Refs.~\cite{ptak1991orthomodular,dvurecenskij2000new}). The first step in
this direction would be identification of how the (2,2)-box world logic arise
from very simple logics of separate boxes.

Authors are not aware of any prior results related to the violation of
Tsirelson bound in the framework of quantum logics. The presented analysis
suggests that other examples of quantum logics studied in the literature can
exhibit violation of Tsirelson bound. These could provide new and interesting
models for quantum information theory. It is also interesting to examine how
compliance with Tsireslon bound for an orthomodular lattice is related to the
property of being a projection lattice of some von Neumann algebra.

\ack

Work of TT was supported by University of Gda\'nsk, grant BMN-538-5400-B168-13 and
within the International PhD Project Physics of future quantum-based
information technologies, grant MPD/2009-3/4 from Foundation for Polish
Science. Part of this publication was made possible through the support of a
grant from the John Templeton Foundation.

Authors are grateful to Micha{\l}‚ Horodecki, W{\l}adys{\l}aw Adam Majewski,
Silvia Pulmannov\'a and Jaros{\l}aw Pykacz for valuable
discussions.



\section*{References}

\bibliographystyle{unsrt}
\bibliography{boxesandlogics}


\bgroup
\begin{table}[h]
\tiny
\begin{tabular}{|l|p{5cm}||l|p{5cm}|}
    \hline
    $q$ & questions that cover $q$ &
    $q$ & questions that cover $q$ \\
    \hline\hline
 $0$ & \multicolumn{3}{|r|}{$\mathrm{[xx,00]},\mathrm{[xx,01]},\mathrm{[xx,10]},\mathrm{[xx,11]}, \mathrm{[xy,00]},$
             $\mathrm{[xy,01]},\mathrm{[xy,10]},\mathrm{[xy,11]},\mathrm{[yx,00]},\mathrm{[yx,01]},$
             $\mathrm{[yx,10]},\mathrm{[yx,11]},\mathrm{[yy,00]},\mathrm{[yy,01]},\mathrm{[yy,10]},$
             $\mathrm{[yy,11]}$} \\
 \hline
 $\mathrm{[xx,00]}$ & $\mathrm{[xx,00]}\oplus \mathrm{[xx,01]},\mathrm{[xx,00]}\oplus \mathrm{[xx,10]},$\newline
                    $\mathrm{[xx,00]}\oplus \mathrm{[xx,11]},\mathrm{[xx,00]}\oplus \mathrm{[xy,10]},$\newline
                    $\mathrm{[xx,00]}\oplus \mathrm{[xy,11]},\mathrm{[xx,00]}\oplus \mathrm{[yx,01]},$\newline
                    $\mathrm{[xx,00]}\oplus \mathrm{[yx,11]}$
                    &
 $\mathrm{[xx,01]}$ & $\mathrm{[xx,00]}\oplus \mathrm{[xx,01]},\mathrm{[xx,01]}\oplus \mathrm{[xx,10]},$\newline
                    $\mathrm{[xx,01]}\oplus \mathrm{[xx,11]},\mathrm{[xx,01]}\oplus \mathrm{[xy,10]},$\newline
                    $\mathrm{[xx,01]}\oplus \mathrm{[xy,11]},\mathrm{[xx,01]}\oplus \mathrm{[yx,00]},$\newline
                    $\mathrm{[xx,01]}\oplus \mathrm{[yx,10]}$ \\
 \hline
 $\mathrm{[xx,10]}$ & $\mathrm{[xx,00]}\oplus \mathrm{[xx,10]},\mathrm{[xx,01]}\oplus \mathrm{[xx,10]},$\newline
                    $\mathrm{[xx,10]}\oplus \mathrm{[xx,11]},\mathrm{[xx,10]}\oplus \mathrm{[xy,00]},$\newline
                    $\mathrm{[xx,10]}\oplus \mathrm{[xy,01]},\mathrm{[xx,10]}\oplus \mathrm{[yx,01]},$\newline
                    $\mathrm{[xx,10]}\oplus \mathrm{[yx,11]}$ &
 $\mathrm{[xx,11]}$ & $\mathrm{[xx,00]}\oplus \mathrm{[xx,11]},\mathrm{[xx,01]}\oplus \mathrm{[xx,11]},$\newline
                    $\mathrm{[xx,10]}\oplus \mathrm{[xx,11]},\mathrm{[xx,11]}\oplus \mathrm{[xy,00]},$\newline
                    $\mathrm{[xx,11]}\oplus \mathrm{[xy,01]},\mathrm{[xx,11]}\oplus \mathrm{[yx,00]},$\newline
                    $\mathrm{[xx,11]}\oplus \mathrm{[yx,10]}$ \\
 \hline
 $\mathrm{[xy,00]}$ & $\mathrm{[xx,00]}\oplus \mathrm{[xx,01]},\mathrm{[xx,10]}\oplus \mathrm{[xy,00]},$\newline
                    $\mathrm{[xx,11]}\oplus \mathrm{[xy,00]},\mathrm{[xy,00]}\oplus \mathrm{[xy,10]},$\newline
                    $\mathrm{[xy,00]}\oplus \mathrm{[xy,11]},\mathrm{[xy,00]}\oplus \mathrm{[yy,01]},$\newline
                    $\mathrm{[xy,00]}\oplus \mathrm{[yy,11]}$ &
 $\mathrm{[xy,01]}$ & $\mathrm{[xx,00]}\oplus \mathrm{[xx,01]},\mathrm{[xx,10]}\oplus \mathrm{[xy,01]},$\newline
                    $\mathrm{[xx,11]}\oplus \mathrm{[xy,01]},\mathrm{[xy,01]}\oplus \mathrm{[xy,10]},$\newline
                    $\mathrm{[xy,01]}\oplus \mathrm{[xy,11]},\mathrm{[xy,01]}\oplus \mathrm{[yy,00]},$\newline
                    $\mathrm{[xy,01]}\oplus \mathrm{[yy,10]}$ \\
 \hline
 $\mathrm{[xy,10]}$ & $\mathrm{[xx,00]}\oplus \mathrm{[xy,10]},\mathrm{[xx,01]}\oplus \mathrm{[xy,10]},$\newline
                    $\mathrm{[xx,10]}\oplus \mathrm{[xx,11]},\mathrm{[xy,00]}\oplus \mathrm{[xy,10]},$\newline
                    $\mathrm{[xy,01]}\oplus \mathrm{[xy,10]},\mathrm{[xy,10]}\oplus \mathrm{[yy,01]},$\newline
                    $\mathrm{[xy,10]}\oplus \mathrm{[yy,11]}$ &
 $\mathrm{[xy,11]}$ & $\mathrm{[xx,00]}\oplus \mathrm{[xy,11]},\mathrm{[xx,01]}\oplus \mathrm{[xy,11]},$\newline
                    $\mathrm{[xx,10]}\oplus \mathrm{[xx,11]},\mathrm{[xy,00]}\oplus \mathrm{[xy,11]},$\newline
                    $\mathrm{[xy,01]}\oplus \mathrm{[xy,11]},\mathrm{[xy,11]}\oplus \mathrm{[yy,00]},$\newline
                    $\mathrm{[xy,11]}\oplus \mathrm{[yy,10]}$ \\
 \hline
 $\mathrm{[yx,00]}$ & $\mathrm{[xx,00]}\oplus \mathrm{[xx,10]},\mathrm{[xx,01]}\oplus \mathrm{[yx,00]},$\newline
                    $\mathrm{[xx,11]}\oplus \mathrm{[yx,00]},\mathrm{[yx,00]}\oplus \mathrm{[yx,01]},$\newline
                    $\mathrm{[yx,00]}\oplus \mathrm{[yx,11]},\mathrm{[yx,00]}\oplus \mathrm{[yy,10]},$\newline
                    $\mathrm{[yx,00]}\oplus \mathrm{[yy,11]}$ &
 $\mathrm{[yx,01]}$ & $\mathrm{[xx,00]}\oplus \mathrm{[yx,01]},\mathrm{[xx,01]}\oplus \mathrm{[xx,11]},$\newline
                    $\mathrm{[xx,10]}\oplus \mathrm{[yx,01]},\mathrm{[yx,00]}\oplus \mathrm{[yx,01]},$\newline
                    $\mathrm{[yx,01]}\oplus \mathrm{[yx,10]},\mathrm{[yx,01]}\oplus \mathrm{[yy,10]},$\newline
                    $\mathrm{[yx,01]}\oplus \mathrm{[yy,11]}$ \\
 \hline
 $\mathrm{[yx,10]}$ & $\mathrm{[xx,00]}\oplus \mathrm{[xx,10]},\mathrm{[xx,01]}\oplus \mathrm{[yx,10]},$\newline
                    $\mathrm{[xx,11]}\oplus \mathrm{[yx,10]},\mathrm{[yx,01]}\oplus \mathrm{[yx,10]},$\newline
                    $\mathrm{[yx,10]}\oplus \mathrm{[yx,11]},\mathrm{[yx,10]}\oplus \mathrm{[yy,00]},$\newline
                    $\mathrm{[yx,10]}\oplus \mathrm{[yy,01]}$ &
 $\mathrm{[yx,11]}$ & $\mathrm{[xx,00]}\oplus \mathrm{[yx,11]},\mathrm{[xx,01]}\oplus \mathrm{[xx,11]},$\newline
                    $\mathrm{[xx,10]}\oplus \mathrm{[yx,11]},\mathrm{[yx,00]}\oplus \mathrm{[yx,11]},$\newline
                    $\mathrm{[yx,10]}\oplus \mathrm{[yx,11]},\mathrm{[yx,11]}\oplus \mathrm{[yy,00]},$\newline
                    $\mathrm{[yx,11]}\oplus \mathrm{[yy,01]}$ \\
 \hline
 $\mathrm{[yy,00]}$ & $\mathrm{[xy,00]}\oplus \mathrm{[xy,10]},\mathrm{[xy,01]}\oplus \mathrm{[yy,00]},$\newline
                    $\mathrm{[xy,11]}\oplus \mathrm{[yy,00]},\mathrm{[yx,00]}\oplus \mathrm{[yx,01]},$\newline
                    $\mathrm{[yx,10]}\oplus \mathrm{[yy,00]},\mathrm{[yx,11]}\oplus \mathrm{[yy,00]},$\newline
                    $\mathrm{[yy,00]}\oplus \mathrm{[yy,11]}$ &
 $\mathrm{[yy,01]}$ & $\mathrm{[xy,00]}\oplus \mathrm{[yy,01]},\mathrm{[xy,01]}\oplus \mathrm{[xy,11]},$\newline
                    $\mathrm{[xy,10]}\oplus \mathrm{[yy,01]},\mathrm{[yx,00]}\oplus \mathrm{[yx,01]},$\newline
                    $\mathrm{[yx,10]}\oplus \mathrm{[yy,01]},\mathrm{[yx,11]}\oplus \mathrm{[yy,01]},$\newline
                    $\mathrm{[yy,01]}\oplus \mathrm{[yy,10]}$ \\
 \hline
 $\mathrm{[yy,10]}$ & $\mathrm{[xy,00]}\oplus \mathrm{[xy,10]},\mathrm{[xy,01]}\oplus \mathrm{[yy,10]},$\newline
                    $\mathrm{[xy,11]}\oplus \mathrm{[yy,10]},\mathrm{[yx,00]}\oplus \mathrm{[yy,10]},$\newline
                    $\mathrm{[yx,01]}\oplus \mathrm{[yy,10]},\mathrm{[yx,10]}\oplus \mathrm{[yx,11]},$\newline
                    $\mathrm{[yy,01]}\oplus \mathrm{[yy,10]}$ &
 $\mathrm{[yy,11]}$ & $\mathrm{[xy,00]}\oplus \mathrm{[yy,11]},\mathrm{[xy,01]}\oplus \mathrm{[xy,11]},$\newline
                    $\mathrm{[xy,10]}\oplus \mathrm{[yy,11]},\mathrm{[yx,00]}\oplus \mathrm{[yy,11]},$\newline
                    $\mathrm{[yx,01]}\oplus \mathrm{[yy,11]},\mathrm{[yx,10]}\oplus \mathrm{[yx,11]},$\newline
                    $\mathrm{[yy,00]}\oplus \mathrm{[yy,11]}$ \\
 \hline
 $\mathrm{[xx,00]}\oplus \mathrm{[xx,01]}$ & $\mathrm{[xx,10]}\cmpl,\mathrm{[xx,11]}\cmpl,
                                          \mathrm{[xy,10]}\cmpl,\mathrm{[xy,11]}\cmpl$ &
 $\mathrm{[xx,00]}\oplus \mathrm{[xx,10]}$ & $\mathrm{[xx,01]}\cmpl,\mathrm{[xx,11]}\cmpl,
                                          \mathrm{[yx,01]}\cmpl,\mathrm{[yx,11]}\cmpl$ \\
 \hline
 $\mathrm{[xx,01]}\oplus \mathrm{[xx,11]}$ & $\mathrm{[xx,00]}\cmpl,\mathrm{[xx,10]}\cmpl,
                                          \mathrm{[yx,00]}\cmpl,\mathrm{[yx,10]}\cmpl$ &
 $\mathrm{[xx,10]}\oplus \mathrm{[xx,11]}$ & $\mathrm{[xx,00]}\cmpl,\mathrm{[xx,01]}\cmpl,
                                          \mathrm{[xy,00]}\cmpl,\mathrm{[xy,01]}\cmpl$ \\
 \hline
 $\mathrm{[xy,00]}\oplus \mathrm{[xy,10]}$ & $\mathrm{[xy,01]}\cmpl,\mathrm{[xy,11]}\cmpl,
                                          \mathrm{[yy,01]}\cmpl,\mathrm{[yy,11]}\cmpl$ &
 $\mathrm{[xy,01]}\oplus \mathrm{[xy,11]}$ & $\mathrm{[xy,00]}\cmpl,\mathrm{[xy,10]}\cmpl,
                                          \mathrm{[yy,00]}\cmpl,\mathrm{[yy,10]}\cmpl$ \\
 \hline
 $\mathrm{[yx,00]}\oplus \mathrm{[yx,01]}$ & $\mathrm{[yx,10]}\cmpl,\mathrm{[yx,11]}\cmpl,
                                          \mathrm{[yy,10]}\cmpl,\mathrm{[yy,11]}\cmpl$ &
 $\mathrm{[yx,10]}\oplus \mathrm{[yx,11]}$ & $\mathrm{[yx,00]}\cmpl,\mathrm{[yx,01]}\cmpl,
                                          \mathrm{[yy,00]}\cmpl,\mathrm{[yy,01]}\cmpl$ \\
 \hline
 $\mathrm{[xx,00]}\oplus \mathrm{[xx,11]}$ & $\mathrm{[xx,01]}\cmpl,\mathrm{[xx,10]}\cmpl$ &
 $\mathrm{[xx,00]}\oplus \mathrm{[xy,10]}$ & $\mathrm{[xx,01]}\cmpl,\mathrm{[xy,11]}\cmpl$ \\
 \hline
 $\mathrm{[xx,00]}\oplus \mathrm{[xy,11]}$ & $\mathrm{[xx,01]}\cmpl,\mathrm{[xy,10]}\cmpl$ &
 $\mathrm{[xx,00]}\oplus \mathrm{[yx,01]}$ & $\mathrm{[xx,10]}\cmpl,\mathrm{[yx,11]}\cmpl$ \\
 \hline
 $\mathrm{[xx,00]}\oplus \mathrm{[yx,11]}$ & $\mathrm{[xx,10]}\cmpl,\mathrm{[yx,01]}\cmpl$ &
 $\mathrm{[xx,01]}\oplus \mathrm{[xx,10]}$ & $\mathrm{[xx,00]}\cmpl,\mathrm{[xx,11]}\cmpl$ \\
 \hline
 $\mathrm{[xx,01]}\oplus \mathrm{[xy,10]}$ & $\mathrm{[xx,00]}\cmpl,\mathrm{[xy,11]}\cmpl$ &
 $\mathrm{[xx,01]}\oplus \mathrm{[xy,11]}$ & $\mathrm{[xx,00]}\cmpl,\mathrm{[xy,10]}\cmpl$ \\
 \hline
 $\mathrm{[xx,01]}\oplus \mathrm{[yx,00]}$ & $\mathrm{[xx,11]}\cmpl,\mathrm{[yx,10]}\cmpl$ &
 $\mathrm{[xx,01]}\oplus \mathrm{[yx,10]}$ & $\mathrm{[xx,11]}\cmpl,\mathrm{[yx,00]}\cmpl$ \\
 \hline
 $\mathrm{[xx,10]}\oplus \mathrm{[xy,00]}$ & $\mathrm{[xx,11]}\cmpl,\mathrm{[xy,01]}\cmpl$ &
 $\mathrm{[xx,10]}\oplus \mathrm{[xy,01]}$ & $\mathrm{[xx,11]}\cmpl,\mathrm{[xy,00]}\cmpl$ \\
 \hline
 $\mathrm{[xx,10]}\oplus \mathrm{[yx,01]}$ & $\mathrm{[xx,00]}\cmpl,\mathrm{[yx,11]}\cmpl$ &
 $\mathrm{[xx,10]}\oplus \mathrm{[yx,11]}$ & $\mathrm{[xx,00]}\cmpl,\mathrm{[yx,01]}\cmpl$ \\
 \hline
 $\mathrm{[xx,11]}\oplus \mathrm{[xy,00]}$ & $\mathrm{[xx,10]}\cmpl,\mathrm{[xy,01]}\cmpl$ &
 $\mathrm{[xx,11]}\oplus \mathrm{[xy,01]}$ & $\mathrm{[xx,10]}\cmpl,\mathrm{[xy,00]}\cmpl$ \\
 \hline
 $\mathrm{[xx,11]}\oplus \mathrm{[yx,00]}$ & $\mathrm{[xx,01]}\cmpl,\mathrm{[yx,10]}\cmpl$ &
 $\mathrm{[xx,11]}\oplus \mathrm{[yx,10]}$ & $\mathrm{[xx,01]}\cmpl,\mathrm{[yx,00]}\cmpl$ \\
 \hline
 $\mathrm{[xy,00]}\oplus \mathrm{[xy,11]}$ & $\mathrm{[xy,01]}\cmpl,\mathrm{[xy,10]}\cmpl$ &
 $\mathrm{[xy,00]}\oplus \mathrm{[yy,01]}$ & $\mathrm{[xy,10]}\cmpl,\mathrm{[yy,11]}\cmpl$ \\
 \hline
 $\mathrm{[xy,00]}\oplus \mathrm{[yy,11]}$ & $\mathrm{[xy,10]}\cmpl,\mathrm{[yy,01]}\cmpl$ &
 $\mathrm{[xy,01]}\oplus \mathrm{[xy,10]}$ & $\mathrm{[xy,00]}\cmpl,\mathrm{[xy,11]}\cmpl$ \\
 \hline
 $\mathrm{[xy,01]}\oplus \mathrm{[yy,00]}$ & $\mathrm{[xy,11]}\cmpl,\mathrm{[yy,10]}\cmpl$ &
 $\mathrm{[xy,01]}\oplus \mathrm{[yy,10]}$ & $\mathrm{[xy,11]}\cmpl,\mathrm{[yy,00]}\cmpl$ \\
 \hline
 $\mathrm{[xy,10]}\oplus \mathrm{[yy,01]}$ & $\mathrm{[xy,00]}\cmpl,\mathrm{[yy,11]}\cmpl$ &
 $\mathrm{[xy,10]}\oplus \mathrm{[yy,11]}$ & $\mathrm{[xy,00]}\cmpl,\mathrm{[yy,01]}\cmpl$ \\
 \hline
 $\mathrm{[xy,11]}\oplus \mathrm{[yy,00]}$ & $\mathrm{[xy,01]}\cmpl,\mathrm{[yy,10]}\cmpl$ &
 $\mathrm{[xy,11]}\oplus \mathrm{[yy,10]}$ & $\mathrm{[xy,01]}\cmpl,\mathrm{[yy,00]}\cmpl$ \\
 \hline
 $\mathrm{[yx,00]}\oplus \mathrm{[yx,11]}$ & $\mathrm{[yx,01]}\cmpl,\mathrm{[yx,10]}\cmpl$ &
 $\mathrm{[yx,00]}\oplus \mathrm{[yy,10]}$ & $\mathrm{[yx,01]}\cmpl,\mathrm{[yy,11]}\cmpl$ \\
 \hline
 $\mathrm{[yx,00]}\oplus \mathrm{[yy,11]}$ & $\mathrm{[yx,01]}\cmpl,\mathrm{[yy,10]}\cmpl$ &
 $\mathrm{[yx,01]}\oplus \mathrm{[yx,10]}$ & $\mathrm{[yx,00]}\cmpl,\mathrm{[yx,11]}\cmpl$ \\
 \hline
 $\mathrm{[yx,01]}\oplus \mathrm{[yy,10]}$ & $\mathrm{[yx,00]}\cmpl,\mathrm{[yy,11]}\cmpl$ &
 $\mathrm{[yx,01]}\oplus \mathrm{[yy,11]}$ & $\mathrm{[yx,00]}\cmpl,\mathrm{[yy,10]}\cmpl$ \\
 \hline
 $\mathrm{[yx,10]}\oplus \mathrm{[yy,00]}$ & $\mathrm{[yx,11]}\cmpl,\mathrm{[yy,01]}\cmpl$ &
 $\mathrm{[yx,10]}\oplus \mathrm{[yy,01]}$ & $\mathrm{[yx,11]}\cmpl,\mathrm{[yy,00]}\cmpl$ \\
 \hline
 $\mathrm{[yx,11]}\oplus \mathrm{[yy,00]}$ & $\mathrm{[yx,10]}\cmpl,\mathrm{[yy,01]}\cmpl$ &
 $\mathrm{[yx,11]}\oplus \mathrm{[yy,01]}$ & $\mathrm{[yx,10]}\cmpl,\mathrm{[yy,00]}\cmpl$ \\
 \hline
 $\mathrm{[yy,00]}\oplus \mathrm{[yy,11]}$ & $\mathrm{[yy,01]}\cmpl,\mathrm{[yy,10]}\cmpl$ &
 $\mathrm{[yy,01]}\oplus \mathrm{[yy,10]}$ & $\mathrm{[yy,00]}\cmpl,\mathrm{[yy,11]}\cmpl$ \\
 \hline
any $\mathrm{[ab,\alpha, \beta]}\cmpl$ & $\bbone$ &
$\bbone$ & $\emptyset$ \\
 \hline\hline
\end{tabular}
    \caption{(2,2)-box world logic structure. Elements covering each question.
        Only single representant is given for each proposition.}
    \label{tab:L}
\pagestyle{empty}
\end{table}
\egroup

\end{document}